\documentclass[letterpaper, 10pt, onecolumn,conference]{ieeeconf}  

\IEEEoverridecommandlockouts                              

\usepackage{graphicx} 
\usepackage{mathptmx} 
\usepackage{times} 
\usepackage{amsmath} 
\usepackage{amssymb}  
\usepackage{datetime}
\usepackage{hyperref}
\usepackage{color}

\newtheorem{thm}{Theorem}

\def\rank{\mathrm{rank}}
\DeclareMathAlphabet{\bit}{OML}{cmm}{b}{it}

\def\<{\leqslant}           
\def\>{\geqslant}           

\def\d{\partial}
\def\wh{\widehat}
\def\wt{\widetilde}
\def\lexp{\mathop{\overleftarrow{\exp}}}
\def\im{\mathrm{im} }   
\def\Re{\mathrm{Re} }   
\def\Im{\mathrm{Im} }   

\def\mR{{\mathbb R}}
\def\mC{\mathbb{C}}

\def\Tr{\mathrm{Tr}}
\def\col{\mathrm{vec}}

\def\rT{\mathrm{T}}

\def\diam{\diamond}

\def\bE{\mathbf{E}}

\def\[[[{[\![\![}
\def\]]]{]\!]\!]}

\def\bra{\langle}
\def\ket{\rangle}

\def\re{\mathrm{e}}
\def\rd{\mathrm{d}}

\def\bJ{\mathbf{J}}

\def\x{\times}
\def\ox{\otimes}

\def\fF{\mathfrak{F}}

\def\fH{\mathfrak{H}}

\def\cI{\mathcal{I}}

\def\cE{\mathcal{E}}

\def\eps{\epsilon}
\def\ups{\upsilon}
\def\Ups{\Upsilon}

\title{\LARGE \bf
Decoherence Time Control by Interconnection for Finite-Level Quantum Memory Systems*}

\author{Igor G. Vladimirov$^{1}$, \qquad Ian R. Petersen$^{2}$
\thanks{*This work is supported by the Australian Research Council  grants  DP210101938, DP200102945.}
\thanks{$^{1,2}$School of Engineering, Australian National University, Acton, ACT 2601, Canberra, Australia, 
        {\tt igor.g.vladimirov@gmail.com, i.r.petersen@gmail.com}.}}

\begin{document}
\maketitle
\thispagestyle{empty}
\pagestyle{empty}

\begin{abstract}
This paper is concerned with open quantum systems whose dynamic variables have an algebraic structure, similar to that of the Pauli matrices for finite-level systems. The Hamiltonian and the operators of coupling of the system to the external bosonic fields depend linearly on the system variables. The fields are represented by quantum Wiener processes which drive the system dynamics according to a quasilinear Hudson-Parthasarathy quantum stochastic differential equation whose drift vector and dispersion matrix are affine and linear functions of the system variables. This setting includes the zero-Hamiltonian isolated system dynamics as a particular case, where the system variables  are constant  in time, which makes them potentially applicable as a quantum memory. In a more realistic case of nonvanishing system-field coupling, we define a  memory decoherence time when a mean-square deviation of the system variables from their initial values becomes relatively  significant as specified by a weighting matrix and a fidelity parameter.  We consider the decoherence time  maximization over the energy parameters of the system and obtain a condition under which the zero Hamiltonian provides a suboptimal solution. This optimization problem is also discussed for a direct energy coupling interconnection of such systems.
\end{abstract}

\section{INTRODUCTION}

The distinctions of quantum mechanics   \cite{LL_1981,S_1994} from classical (deterministic and stochastic) dynamics, exploited as potential resources in the development of quantum information and quantum computation technologies \cite{NC_2000}, come from the noncommutative operator-valued nature of quantum dynamic variables and quantum probability (where classical probability measures are replaced with quantum states in the form of density operators \cite{H_2001}). In the case of finite-level systems, which are particularly relevant to the above applications,  quantum mechanical operators act on a finite-dimensional complex Hilbert space and are organised as square matrices, with Hermitian matrices representing real-valued physical quantities.  If the quantum system variables (together with the identity matrix) form a basis  in the appropriate matrix space (as exemplified by the three Pauli $(2\x 2)$-matrices \cite{S_1994} as self-adjoint operators on the qubit space $\mC^2$), their pairwise products are affine functions of  the same set of variables, thus reducing any function (for example, a polynomial) of such variables to an affine function. This algebraic structure (along with the structure constants)   is preserved over the course of time regardless of whether the quantum system is isolated from its environment or is open and interacts with the external fields and other systems.

In the case of open quantum dynamics, the energetics of the system and its interaction with the surroundings is specified by the system Hamiltonian and the operators of coupling of the system to the external fields. In the presence of the algebraic structure,  the Hamiltonian and the coupling operators are, without loss of generality, linear and affine functions of the system variables,  parameterized by an energy vector and coupling parameters. In the framework of the Hudson-Parthasarathy calculus \cite{HP_1984,P_1992}, this leads to  quasilinear quantum stochastic differential equations (QSDEs) \cite{EMPUJ_2016,VP_2022_SIAM} whose drift vector depends affinely and the dispersion matrix depends linearly on the system variables, thereby resembling  the classical SDEs with state-dependent noise \cite{W_1967}.

In particular, if the energy and coupling parameters vanish, so that the quantum system is isolated and has zero Hamiltonian, then not only the Hamiltonian, but  all the system variables are conserved in time, thus forming a set of (in general noncommuting) quantum variables which store their  initial values. However, because of the unavoidable system-field coupling,  the quantum memory property can hold only approximately and over a bounded time interval. The energy exchange between the system and its environment  produces quantum noise \cite{GZ_2004} (accompanied by decoherence effects; see for example, \cite{VP_2023_SCL} and references  therein),  which makes the system variables drift away from their initial conditions even if the energy vector (and hence, the Hamiltonian) is zero.
The  time when this deviation 
becomes relatively significant  on a ``typical'' scale of the initial system variables suggests a performance index for the system as a temporary quantum memory  (see \cite{FCHJ_2016,YJ_2014} and references therein for other settings).

In the present paper, we use the tractability of the second and higher-order moment dynamics for the system variables \cite{VP_2022_SIAM} (in the case of external fields in the vacuum state \cite{P_1992}) and the fundamental solutions of quasilinear QSDEs in the form of  time-ordered exponentials. Similarly to \cite{VP_2023_ANZCC}, this leads to a practically computable memory decoherence time,  which is defined in terms of a weighted mean-square deviation of the system variables from their  initial values and involves a weighting matrix along with a dimensionless fidelity parameter.
We propose a problem of maximizing the decoherence time over the energy and coupling parameters of the system and discuss its approximate version using a quadratically truncated Taylor series expansion of the decoherence time over the fidelity parameter. With respect to the energy vector,  this approximate decoherence time maximization is organized as a quadratic optimization problem which admits a closed-form solution. The latter provides  a condition on the coupling parameters under which the zero energy vector is a suboptimal solution of the decoherence time maximization  problem. 
We apply this optimization approach to an interconnection of such systems through a  direct energy coupling.

\section{ISOLATED QUANTUM SYSTEM WITH AN ALGEBRAIC STRUCTURE}
\label{sec:iso}

In the Heisenberg picture of quantum dynamics \cite{S_1994}, physical quantities are represented by quantum variables $\xi(t)$, which are linear operators acting on a Hilbert space $\fH$ and depending on time $t\> 0$, while  the quantum probabilistic structure is specified by a fixed  positive semi-definite self-adjoint   density operator (a quantum state) $\rho = \rho^\dagger \succcurlyeq 0$ of unit trace $\Tr \rho = 1$ on $\fH$ along with the expectation $\bE \xi := \Tr (\rho \xi)$.
As in \cite{VP_2022_SIAM}, we consider a quantum system with dynamic variables $X_1(t), \ldots, X_n(t)$ which are self-adjoint operators on $\fH$, assembled into a column-vector  $X:= (X_k)_{1\< k\< n}$ with an algebraic structure
\begin{equation}
\label{Xalg}
    X(t)X(t)^\rT = \alpha  + \sum_{\ell=1}^n \beta_\ell X_\ell (t),
    \qquad
    t \> 0,
\end{equation}
whereby every pairwise product $X_jX_k$  is an affine function of $X_1, \ldots, X_n$.
Here, $\alpha := (\alpha_{jk})_{1\< j,k\< n}= \alpha^\rT \in \mR^{n\x n}$ and $\beta_\ell:= (\beta_{jk\ell})_{1\< j,k\< n} = \beta_\ell^* \in \mC^{n\x n}$ are constant real symmetric and complex Hermitian matrices, respectively, with $(\cdot)^*:= \overline{(\cdot)}{}^\rT$  the complex conjugate transpose.  In (\ref{Xalg}),      the matrix $\alpha$ is identified with  its tensor product $\alpha \ox \cI = (\alpha_{jk}\cI)_{1\< j,k\< n}$ with the identity operator $\cI$ on the space $\fH$, and, in a similar fashion,  $\beta_\ell X_\ell:= (\beta_{jk\ell}X_\ell)_{1\< j,k\< n}$ is an $(n\x n)$-matrix of operators which are ``rescaled''  copies of $X_\ell$. The matrices $\beta_1, \ldots, \beta_n$ are sections of an array $\beta:= (\beta_{jk\ell})_{1\< j,k,\ell\< n} \in \mC^{n\x n\x n}$. Their imaginary parts are real antisymmetric matrices
\begin{equation}
\label{Thetaell}
    \Theta_\ell
    :=
    (\theta_{jk\ell})_{1\< j,k\< n}
    :=
    \Im \beta_\ell
    =
    -\Theta_\ell^\rT
    \in
    \mR^{n\x n}
\end{equation}
which form the corresponding sections of the  $(n\x n\x n)$-array
\begin{equation}
\label{Theta}
    \Theta : =  (\theta_{jk\ell})_{1\< j,k,\ell \< n}:=  \Im \beta \in \mR^{n\x n\x n}.
\end{equation}
The latter specifies the  canonical commutation relations (CCRs) for the system variables (at the same but otherwise arbitrary moment of time):
\begin{equation}
\label{XXcomm}
    [X,X^\rT]
      :=
    ([X_j,X_k])_{1\< j,k\< n}
    =
    2i \Theta \cdot X,
\end{equation}
where $[\xi,\eta]:= \xi\eta -\eta\xi$ is the commutator of linear operators.
Here, for any array $\gamma:= (\gamma_{jk\ell})_{1\< j,k,\ell \< n} \in \mC^{n\x n\x n}$ with sections $\gamma_\ell:= (\gamma_{jk\ell})_{1\< j,k \< n} \in \mC^{n\x n}$ and $\gamma_{\bullet k \bullet}:= (\gamma_{jk\ell})_{1\< j,\ell \< n} \in \mC^{n\x n}$,  and any vector $u:= (u_\ell)_{1\< \ell \< n}$ of $n$ quantum variables on $\fH$, we use the following  product:
\begin{equation}
\label{cdot}
    \gamma \cdot u
    :=
    \sum_{\ell = 1}^n
    \gamma_\ell  u_\ell
    =
    \begin{bmatrix}
      \gamma_{\bullet 1 \bullet} u
      &
      \ldots &
      \gamma_{\bullet n \bullet} u
    \end{bmatrix},
\end{equation}
which yields an $(n\x n)$-matrix of quantum variables, with the columns
$\gamma_{\bullet k \bullet} u$.    Accordingly, the rightmost sum in (\ref{Xalg}) can be represented as $\beta\cdot X$. We will also employ a different product 
\begin{equation}
\label{diam}
    \gamma \diam u
    :=
    \begin{bmatrix}
        \gamma_1 u & \ldots & \gamma_n u
    \end{bmatrix},
\end{equation}
which,  as mentioned in \cite{VP_2022_SIAM}, is associated with (\ref{cdot}) by
\begin{equation}
\label{cdotdiam}
  (\gamma \cdot u)v
  =
  (\gamma\diam v)u
  =
  \begin{bmatrix}
    \gamma_1 & \ldots &  \gamma_n
  \end{bmatrix}
  (u\ox v)
\end{equation}
(with $\ox$ the Kronecker product)
for any vectors $u$, $v$  of $n$ quantum variables with zero cross-commutations: $[u,v^\rT] = 0$.

An example of $n=3$ quantum variables with
the algebraic structure (\ref{Xalg}) is provided by the Pauli matrices \cite{S_1994}
\begin{equation}
\label{X123}
    \sigma_1:=
    {\begin{bmatrix}
      0 & 1\\
      1 & 0
    \end{bmatrix}},
    \qquad
    \sigma_2:=
    {\begin{bmatrix}
      0 & -i\\
      i & 0
    \end{bmatrix}}
    =
    -i\bJ,
    \qquad
    \sigma_3:=
    {\begin{bmatrix}
      1 & 0\\
      0 & -1
    \end{bmatrix}},
\end{equation}
which are traceless Hermitian matrices on the Hilbert space $\fH:= \mC^2$ of a qubit as the simplest finite-level quantum system,  where
\begin{equation}
\label{bJ}
        \bJ
        : =
        {\begin{bmatrix}
        0 & 1\\
        -1 & 0
    \end{bmatrix}}
\end{equation}
spans the subspace of antisymmetric $(2\x 2)$-matrices.
The structure constants for (\ref{X123}) form the identity matrix $\alpha$ of order $3$
and an  imaginary $(3\x 3\x 3)$-array $\beta$:
\begin{equation}
\label{alfbet}
    \alpha = I_3,
    \qquad
    \beta = i \Theta.
\end{equation}
Here, in accordance with (\ref{Theta}), the CCR array $\Theta \in\{0, \pm1\}^{3\x 3\x 3}$ consists of the Levi-Civita symbols  $\theta_{jk\ell} = \eps_{jk\ell}$, and its sections
\begin{equation}
\label{T123}
    \Theta_1 =
    {\begin{bmatrix}
     0&      0  &   0\\
     0&     0   &  1\\
     0&    -1   &  0
    \end{bmatrix}},
    \qquad
    \Theta_2 =
    {\begin{bmatrix}
     0&      0  &   -1\\
     0&     0   &  0\\
     1&    0   &  0
    \end{bmatrix}},
    \qquad
    \Theta_3 =
    {\begin{bmatrix}
     0&      1  &   0\\
     -1&     0   &  0\\
     0&    0   &  0
    \end{bmatrix}}
\end{equation}
in (\ref{Thetaell}) form
a basis in the subspace of antisymmetric $(3\x 3)$-matrices,   with $(\Theta\diam u)v$ being the cross product of vectors $u,v\in \mR^3$, where 
    $\Theta \diam u
    =
    \begin{bmatrix}
      \Theta_1 u &
      \Theta_2 u &
    \Theta_3 u
    \end{bmatrix}
    =
    {\scriptsize\begin{bmatrix}
      0 & -u_3 & u_2\\
      u_3 & 0 & -u_1\\
      -u_2 & u_1 & 0
    \end{bmatrix}}$ 
is the  infinitesimal generator for the group of rotations in $\mR^3$ about $u$ (as an angular velocity vector)
in view of (\ref{diam}), (\ref{T123}). The set $\{I_2, \sigma_1, \sigma_2, \sigma_3\}$ is a  basis  in the four-dimensional real space of complex Hermitian $(2\x 2)$-matrices which describe self-adjoint operators on the qubit space $\mC^2$.

Returning to the general case of quantum system variables satisfying (\ref{Xalg}), we note that their algebraic structure reduces any polynomial of the system variables to an affine function. Therefore, the Hamiltonian, which determines the  internal energy of the quantum system 
\cite{S_1994},  is assumed, without  loss of generality, to be a linear function of the system variables:
\begin{equation}
\label{HE}
    H
    :=
    E^\rT X,
\end{equation}
where $E\in \mR^n$ is an energy vector.  As in the classical case \cite{A_1989}, adding a constant $c\cI$, with $c\in \mR$,  to $H$ is redundant. The Hamiltonian $H$ completely specifies the Heisenberg  evolution of the isolated quantum system according to the ODE
\begin{equation}
\label{Xdot}
  \dot{X} =
  i[H,X]
  =
  -i[X, X^\rT] E
  =
  2(\Theta \cdot X)E
  =
  A_0 X,
\end{equation}
where $\dot{(\ )}:= \rd/\rd t$ is the time derivative. Here, the matrix $A_0 \in \mR^{n\x n}$ results  from a combination of (\ref{HE}) with  (\ref{XXcomm})--(\ref{cdotdiam}):
\begin{equation}
\label{A0}
    A_0
    := 2\Theta\diam E =
    2\begin{bmatrix}
      \Theta_1 E & \ldots & \Theta_n E
    \end{bmatrix},
\end{equation}
and its linear dependence on $E$ is represented in vectorized form \cite{M_1988} as \begin{equation}
\label{colA0}
    \col(A_0) = 2\mho E,
    \qquad
    \mho
    :=
    {\begin{bmatrix}
        \Theta_1\\
        \vdots\\
        \Theta_n
    \end{bmatrix}}.
\end{equation}
Due to the antisymmetry of the matrices $\Theta_1, \ldots, \Theta_n$ in (\ref{Thetaell}),
the columns $2\Theta_k E$  of $A_0$ in (\ref{A0}) are orthogonal to the energy vector $E$, so that
\begin{equation}
\label{EA0}
    E^\rT  A_0 =      2
    \begin{bmatrix}
        E^\rT\Theta_1 E & \ldots & E^\rT\Theta_n E
    \end{bmatrix}
 = 0.
\end{equation}
The latter leads to $\dot{H} = E^\rT \dot{X} = E^\rT A_0 X = 0$ which is a manifestation  of the property $\dot{H} = i[H,H] = 0$ that the Hamiltonian of an isolated system is a conserved operator regardless of a particular dependence of $H$ on $X$.
Furthermore, if $E=0$ (that is, when  the system Hamiltonian (\ref{HE}) vanishes),  not only $H$ is preserved in time, but so also are the system variables $X_1, \ldots, X_n$ themselves.
The conservation of a set of noncommutative quantum  variables (rather than a single Hamiltonian) can be regarded as a quantum mechanical resource  for making such a system potentially applicable as a quantum memory. A less restrictive condition is provided by a nondecaying dependence of $X(t)$ on $X(0)$, as $t \to +\infty$, considered below.

From (\ref{EA0}),  it follows that  
$
    E \in \ker (A_0^{\rT})
$. Hence, if $E\ne 0$, then
$E$ is an eigenvector of $A_0^\rT$ with zero eigenvalue, which implies that $\det A_0=0$.  Furthermore, the eigenvectors of $A_0$ with nonzero eigenvalues can only be in the  hyperplane with the normal $E$  since
$
    \im A_0:=
    A_0 \mC^n  \subset E^\bot := \{v \in \mC^n:\ E^\rT v = 0\}
$.
By \cite[Theorem 5.1]{VP_2022}, if the matrix $\alpha$ in (\ref{Xalg}) is positive definite 
(as exemplified by (\ref{alfbet}) for the Pauli matrices (\ref{X123})), then for any energy vector $E \in \mR^n$, the matrix $A_0$ in (\ref{A0}) is diagonalizable and its spectrum is on the imaginary axis $i\mR$ and symmetric about the origin (moreover, $\det A_0=0$ with necessity if the dimension $n$ is odd).  In this case,  the system variables perform oscillatory motions or are static (which corresponds to purely imaginary or zero eigenvalues of $A_0$, respectively) with no dissipation because the dependence of $X(t) = \re^{tA_0} X(0)$ on the initial condition $X(0)$ from (\ref{Xdot}) does not fade away, as $t   \to +\infty$, and $X(t)$ does not ``lose memory'' about  $X(0)$.

The issues of ``initializing'' and ``retrieving''   the system variables, which correspond   to the ``write'' and ``read'' memory operations,  are beyond the scope of this paper. Instead, we will discuss the effect of coupling (usually, dissipative)  between the system and its environment (including other quantum systems and  external fields), which forces the system variables evolve in time even if the Hamiltonian is zero.

\section{OPEN QUANTUM SYSTEM}
\label{sec:open}

Compared to the isolated quantum system dynamics (\ref{Xdot}), a more realistic setting is provided by an open version of the system shown schematically in Fig.~\ref{fig:open}.
The internal and output variables of the open quantum system evolve in time according to the Hudson-Parthasarathy 
QSDEs  \cite{EMPUJ_2016,VP_2022_SIAM}:
\begin{equation}
\label{dXdY}
    \rd X  = (AX+b) \rd t + B(X)\rd W,
    \qquad
    \rd Y  = (CX+d) \rd t + D\rd W,
\end{equation}
the first of which is quasi-linear, while the second one is linear.
Here, $A \in \mR^{n\x n}$, $b \in \mR^n$, $C \in \mR^{r\x n}$, $d \in \mR^r$,  $D\in \mR^{r\x m}$  are constant matrices and vectors, with $r\< m$ even, while $B(X)$ is an $(n\x m)$-matrix  of self-adjoint operators, which depend linearly on the system variables as specified below.
\begin{figure}[htbp]
\centering
\unitlength=1mm
\linethickness{0.5pt}
\begin{picture}(40.00,5)
    \put(15,-3){\framebox(10,6)[cc]{\scriptsize system}}
    \put(35,0){\vector(-1,0){10}}
    \put(15,0){\vector(-1,0){10}}
    \put(38,0){\makebox(0,0)[cc]{\small $W$}}
    \put(2,0){\makebox(0,0)[cc]{\small $Y$}}
\end{picture}
\caption{An open quantum system with the input and output fields $W$, $Y$ governed by (\ref{dXdY}).}
\label{fig:open}
\end{figure}
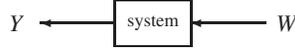
The QSDEs (\ref{dXdY}) are driven by a vector $W:=(W_k)_{1\< k \< m}$ of quantum Wiener processes which are time-varying self-adjoint operators on a symmetric Fock space \cite{PS_1972} $\fF$ representing the input bosonic fields. In contrast to the components of the  classical Brownian motion \cite{KS_1991} in $\mR^m$, the quantum processes $W_1, \ldots, W_m$  do not commute with each other and  have a complex positive semi-definite Hermitian Ito matrix  $\Omega = \Omega^* \succcurlyeq 0$: 
\begin{equation}
\label{Omega}
    \rd W \rd W^\rT
    = \Omega \rd t,
    \qquad
    \Omega: = I_m + iJ,
    \qquad
    J
    :=
    I_{m/2}\ox\bJ,
\end{equation}
where
$\bJ$ is given by (\ref{bJ}). The matrix $J= \Im \Omega$ specifies the two-point CCRs $  [W(s), W(t)^\rT]
  =
  2i \min(s,t)J
$ for the quantum Wiener process  $W$ for any
  $ s, t \> 0$.
Also, $Y:= (Y_k)_{1\< k\< r}$ in (\ref{dXdY}) is a vector of $r$ time-varying self-adjoint operators on the space $\fH$,   which are selected from the  full set of the output fields produced by the interaction of the system with the input fields. The feedthrough matrix $D$ consists of conjugate pairs of  $r$ rows  of a permutation $(m\x m)$-matrix, so that, without loss of generality,  $Y$ has the quantum Ito matrix $D\Omega D^\rT = I_r + i I_{r/2} \ox \bJ$ in view of (\ref{Omega}) and the Ito product rules (in particular, $Y$ includes all the output fields if $r=m$ and $D=I_m$).
The corresponding system-field space
\begin{equation}
\label{fH}
    \fH := \fH_0 \ox \fF
\end{equation}
is the tensor product of the initial system space $\fH_0$ (for the action of  $X_1(0), \ldots, X_n(0)$) and the Fock space $\fF$. Accordingly, the system-field quantum state on $\fH$ is assumed to be in the form
\begin{equation}
\label{rho}
    \rho:= \rho_0\ox \ups,
\end{equation}
where $\rho_0$ is the initial system state on $\fH_0$, and $\ups$ is the vacuum field state \cite{P_1992} on $\fF$.
  Whereas the internal energy of the open quantum system is specified by the Hamiltonian $H$ in  (\ref{HE}) as before, the system-field coupling (through the energy exchange between the system and the external quantum fields) is captured by a vector
  \begin{equation}
  \label{LMN}
    L:= MX + N
  \end{equation}
  of $m$ self-adjoint coupling operators,  which are affine functions of the system variables parameterized by $M \in \mR^{m\x n}$, $N\in \mR^m$. The coefficients    of the QSDEs (\ref{dXdY}) are computed in terms of the structure constants from (\ref{Xalg}), the CCR matrices in (\ref{Thetaell}), (\ref{Omega}) and the energy and coupling parameters in (\ref{HE}), (\ref{LMN}) as follows (see \cite[Theorem 3.1]{VP_2022_SIAM} and \cite[Lemma~4.2 and Theorem~6.1]{EMPUJ_2016}):
\begin{align}
\label{A_b}
    A
    & :=
    A_0 + \wt{A},
    \qquad
    b
      :=
    -2\mho^\rT \col(M^\rT J M \alpha),\\
\label{BX}
    B(X)
     & := 2(\Theta \cdot X)M^\rT,\\
\label{C_d}
    C
    & := 2DJ M,
    \qquad
    d:= 2DJ N.
\end{align}
Here, use is also made of the products  (\ref{cdot}), (\ref{diam})  for the CCR array $\Theta$ from (\ref{Theta}) with the vector $X$ of system variables  and $E + M^\rT J N \in \mR^n$, along with the isolated dynamics matrix $A_0$ in (\ref{A0}) and the auxiliary matrix $\mho$ from (\ref{colA0}). The remaining part of the matrix $A$ in (\ref{A_b}) is given by
\begin{equation}
\label{At}
    \wt{A}
    :=
        2
        \Theta \diam (M^\rT JN)
    +
    2
    \sum_{\ell = 1}^n
    \Theta_\ell
    M^\rT
    (
        M\theta_{\ell\bullet \bullet}
        +
        J M\Re \beta_{\ell\bullet \bullet}),
\end{equation}
which is a quadratic function of the coupling parameters $M$, $N$. In view of (\ref{A0}), (\ref{At}), 
for any fixed $M$, $N$, the matrix $A$ in (\ref{A_b}) belongs to a proper (whenever $n>1$) affine subspace $\im(\Theta \diam) + \wt{A} \subset \mR^{n\x n}$ since the dimension of the subspace $\im(\Theta \diam) :=  \Theta \diam \mR^n = \{\Theta \diam u:\ u \in \mR^n\}$ does not exceed $n$.

If the system is uncoupled from the environment (that is, $M = 0$, $N=0$ in (\ref{LMN})), then $\wt{A}$, $b$,  $B(X)$, $C$, $d$ in (\ref{A_b})--(\ref{At})  vanish. In this case, $A=A_0$ as given by (\ref{A0}),  the first QSDE in (\ref{dXdY}) reduces to the ODE (\ref{Xdot}), and the second QSDE takes the form $\rd Y = D\rd W$, whereby the system has no effect on the output field $Y$. Also note that in the presence of system-field coupling, the matrix $A$ in (\ref{A_b}), which depends on the coupling  parameters $M$, $N$  through $\wt{A}$ in (\ref{At}),   can be made nonzero even in the zero-Hamiltonian case of $E=0$ (when $A_0 = 0$).

According to \cite[Theorem 3.2]{VP_2022_SIAM}, the system variables, as a solution of the quasi-linear QSDE in (\ref{dXdY}),  can be represented as
\begin{equation}
\label{Xts}
    X(t) = \cE(t,0)X(0) + \int_0^t \cE(t,s)\rd s b,
    \qquad
    t \> 0,
\end{equation}
in terms of an $(n\x n)$-matrix  $\cE(t,s)
  :=
  (\cE_{jk}(t,s))_{1\< j,k\< n}
$ of self-adjoint operators on the Fock space $\fF$ defined for any $t\> s\> 0$ by the leftwards time-ordered operator exponential
\begin{equation}
\label{Ets}
  \cE(t,s)
  :=
    \lexp
  \int_s^t
  (A\rd \tau + 2\Theta \diam (M^\rT \rd W(\tau))).
\end{equation}
The columns $\cE_{\bullet k}(t,s):= (\cE_{jk}(t,s))_{1\< j \< n}$ of the matrix $\cE(t,s)$ form the fundamental solutions of the  homogeneous QSDE
\begin{equation}
\label{dEts}
    \rd_t \cE_{\bullet k}(t,s)
    =
    A \cE_{\bullet k}(t,s)\rd t + B(\cE_{\bullet k}(t,s))\rd W(t)
    =
    (A\rd t + 2\Theta \diam (M^\rT \rd W(t)))\cE_{\bullet k}(t,s),
\end{equation}
which is obtained by representing the first QSDE in (\ref{dXdY}) as $\rd X = (A\rd t + 2\Theta \diam (M^\rT \rd W))X + b \rd t$ (in view of (\ref{cdotdiam}) and the commutativity $[X, \rd W^\rT] =0$ between the system variables and the  future-pointing Ito increments of $W$) and    removing the vector $b$.   For any $k = 1, \ldots, n$, the QSDE (\ref{dEts}) is initialized at the $k$th standard basis vector  in $\mR^n$ as $\cE_{\bullet k}(s,s) = (\delta_{jk})_{1\< j\< n}$, with $\delta_{jk}$ the Kronecker delta,  so that   $\cE(s,s) =I_n$. Due to the continuous tensor-product structure of the Fock space $\fF$ and the  vacuum field state assumption in (\ref{rho}),  the quantum variables $\cE_{jk}(t,s)$ commute with and are statistically independent  of $X_\ell(s)$ for all $j,k,\ell = 1, \ldots, n$ and  $t\> s\> 0$.
In the case of vacuum input fields being considered, the diffusion term $B(\cE_{\bullet k})\rd W$ is a martingale part  of the QSDE in (\ref{dEts}) and does not contribute to its averaging, which leads to
\begin{equation}
\label{EEts}
    \bE \cE(t,s) = \re^{(t-s)A},
    \qquad
    t \> s \> 0
\end{equation}
(see \cite[Eq. (4.20)]{VP_2022_SIAM}). This suggests that the dependence of $X(t)$ in (\ref{Xts}) on the initial condition $X(0)$ through the term $\cE(t,0) X(0)$ decays exponentially fast, as $t\to +\infty$, if the matrix $A$ is Hurwitz. In the latter case, there is a limit
\begin{equation}
\label{mu*}
    \mu_\infty:=
  \lim_{t\to +\infty}\mu(t) = -A^{-1} b
\end{equation}
for the mean vector of the system variables, which is computed (regardless of the spectral properties of $A$) by averaging (\ref{Xts}) and using (\ref{EEts}) as
\begin{equation}
\label{mu}
  \mu(t)
  :=
  \bE X(t) = \re^{tA}\mu(0) + \psi(t) b.
\end{equation}
Here, use has also been made of an $\mR^{n\x n}$-valued function
\begin{equation}
\label{psi}
    \psi(t):= \int_0^t \re^{sA}\rd s
    =
    \sum_{k=1}^{+\infty}
    \frac{t^k}{k!}
     A^{k-1}
     =
     (\re^{tA}-I_n)A^{-1},
\end{equation}
which results from
evaluating  \cite{H_2008} the entire function $\frac{\re^{tz}-1}{z}$ of $z \in \mC$ (extended to $t$ at $z=0$ by continuity) at the matrix $A$ in (\ref{A_b}). The rightmost equality in (\ref{psi}) applies to the case of $\det A\ne0$.

The matrix exponential $\re^{tA}$ is also present in the following alternative form of the solution of the first QSDE in (\ref{dXdY}) (with $b\rd t + B(X)\rd W$ being regarded as a forcing term added to the right-hand side of the homogeneous QSDE $\rd X = A X \rd t$):
\begin{equation}
\label{Xsol}
    X(t) = \re^{tA} X(0) + \psi(t) b + Z(t).
\end{equation}
Here, due to the input field $W$ being in the vacuum state $\ups$ on  $\fF$  (and regardless of whether the matrix $A$ is Hurwitz),
the response
\begin{equation}
\label{Z}
    Z(t)
    :=
    \int_0^t
    \re^{(t-s)A}
    B(X(s))
    \rd W(s)
\end{equation}
of the system variables to $W$ is a zero-mean quantum process on the system-field space $\fH$ in (\ref{fH}), uncorrelated with $X(0)$:
\begin{equation}
\label{EZX0}
    \bE Z(t) = 0,
    \qquad
  \bE(Z(t)X(0)^\rT) = 0.
\end{equation}
However, $X(t)$ in (\ref{Xsol}) depends on $X(0)$ not only through the term $\re^{tA} X(0)$ but also through the process $Z(t)$, whereas in (\ref{Xts}), the dependence $X(0)\mapsto X(t)$ is captured in $\cE(t,0)X(0)$ since $\cE(t,s)$ in (\ref{Ets}) is independent of $X(0)$ for any $t\> s\> 0$. By comparing (\ref{Xts}) with (\ref{Xsol}) and using (\ref{EEts}), (\ref{psi}), the process $Z$ in (\ref{Z}) can be  expressed as
\begin{equation}
\label{ZE}
    Z(t) = \wt{\cE}(t,0) X(0) +\int_0^t \wt{\cE}(t,s) \rd s b
\end{equation}
in terms of centered versions of the  exponentials $\cE$ in (\ref{Ets}):
\begin{equation}
\label{Et}
    \wt{\cE}(t,s): = \cE(t,s) - \re^{(t-s)A},
    \qquad
    \bE  \wt{\cE}(t,s) = 0.
\end{equation}
Note that (\ref{EZX0}) can also be obtained from (\ref{ZE}), (\ref{Et}) by using the commutation and statistical properties of (\ref{Ets}) mentioned above.

With the view of employing the system as a quantum memory (aiming to retain its initial variables without dissipation),   of particular interest are those energy and coupling parameters $E$, $M$, $N$ in (\ref{HE}), (\ref{LMN})  for which at least some of the eigenvalues of the matrix $A$ in (\ref{A_b}) are zero or purely imaginary.
On the other hand, in the absence of dissipation (that is, when the matrix $A$ is not Hurwitz),  the processes  $Z(t)$ in (\ref{Xsol}) or $\cE(t,s)$ in (\ref{Xts}), which play the role of quantum noises contaminating the    dependence of $X(t)$ on $X(0)$, grow with time $t$.

These qualitatively different scenarios (when $A$ is Hurwitz and when its spectrum is imaginary) can be analyzed in a unified fashion
in terms of a signal-to-noise ratio which quantifies the relative size of a useful part of $X(t)$ carrying information about $X(0)$.

\section{MEAN-SQUARE DEVIATION FROM INITIAL CONDITIONS}
\label{sec:dev}

By (\ref{Xsol}), the deviation of the system variables at time $t\> 0$ from their initial values takes the form
\begin{equation}
\label{xi}
  \xi(t)
   :=
  X(t)-X(0)
  =
  (\re^{tA}-I_n)X(0) + \psi(t)b + Z(t).
\end{equation}
Its typical ``size''
can be quantified by a weighted mean-square functional
\begin{equation}
\label{Del}
    \Delta(t)
     :=
    \bE(\xi(t)^\rT \Sigma\xi(t))
    =
    \bra \Sigma,  \Re \Ups(t)\ket,
\end{equation}
where $0 \preccurlyeq\Sigma = \Sigma^\rT \in \mR^{n\x n}$ is a given weighting matrix which specifies the relative importance of the system variables, and  $\bra\cdot, \cdot\ket $ is   the Frobenius inner product generating the Frobenius norm $\|\cdot\| $ of matrices \cite{HJ_2007}. Here, use is made of the second-moment matrix
\begin{equation}
\label{Ups}
    \Ups(t)
    :=
    \bE (\xi(t)\xi(t)^\rT)
    = \Ups(t)^* \succcurlyeq 0,
\end{equation}
which enters (\ref{Del}) only through its real part because of the orthogonality between the subspaces of symmetric and antisymmetric matrices, whereby $\bra \Sigma, \Im \Ups(t)\ket  = 0$.  Without loss of generality, the weighting matrix can be  factorized as
\begin{equation}
\label{FF}
    \Sigma := F^\rT F,
    \qquad
    F \in \mR^{\nu\x n},
    \qquad
    \nu:= \rank \Sigma \< n,
\end{equation}
so that $\Delta(t)$ in  (\ref{Del}) ``penalizes'' $\nu$ independent linear combinations of the system variables of interest specified by real coefficients comprising the rows of the full row rank matrix $F$.
By combining (\ref{EZX0}) with  (\ref{xi}), the matrix (\ref{Ups}) takes the form
\begin{align}
\nonumber
    \Ups(t)
    = &
    (\re^{tA}-I_n)
    \Pi
    (\re^{tA^\rT}-I_n)\\
\nonumber
    & +
    (\re^{tA}-I_n)\mu(0)b^\rT \psi(t)^\rT
    +
    \psi(t)b \mu(0)^\rT (\re^{tA^\rT}-I_n)\\
\label{Exixi}
    & +
    \psi(t)bb^\rT \psi(t)^\rT +
    V(t).
\end{align}
Here, $\Pi$ is the second-moment matrix of the initial system variables, computed by averaging both sides of (\ref{Xalg}) and using (\ref{Theta})--(\ref{cdot}) along with (\ref{mu}) as
\begin{align}
\label{EXX0}
    \Pi
    & :=
    \bE (X(0)X(0)^\rT)
    =
    P + i\Theta\cdot \mu(0),\\
\label{P}
    P
    & := \Re \Pi = \alpha + \Re \beta\cdot \mu(0).
\end{align}
In (\ref{Exixi}),  use has also been made of the second-moment matrix
\begin{equation}
\label{V}
    V(t)
      :=
    \bE(Z(t)Z(t)^\rT)
    =
    \int_0^t
    \re^{(t-s)A}
    \Lambda(s)
    \re^{(t-s)A^\rT}
    \rd s,
\end{equation}
where $\Lambda$ describes the averaged quantum Ito matrix for the diffusion term $B(X)\rd W$ in (\ref{dXdY}) and is computed by combining (\ref{BX}) with the antisymmetry of the matrices (\ref{Thetaell}):
\begin{align}
\nonumber
    \Lambda(t)
    & := \bE(B(X(t))\Omega B(X(t))^\rT)\\
\nonumber
    & =
    -4
    \sum_{j,k=1}^n
    \bE
    (X_j(t)X_k(t))
    \Theta_j M^\rT\Omega M\Theta_k\\
\label{EBB}
    & =
    4\mho^\rT ((\alpha + \beta\cdot \mu(t))\ox  (M^\rT\Omega M))\mho
\end{align}
(cf. \cite[Eq. (4.38)]{VP_2022_SIAM}), with
$\Omega$ the quantum Ito matrix of $W$ in (\ref{Omega}) and the matrix $\mho$ from (\ref{colA0}). The function $V$  in (\ref{V})
satisfies the Lyapunov ODE
\begin{equation}
\label{VODE}
    \dot{V}(t) = AV(t) + V(t)A^\rT + \Lambda(t),
\end{equation}
with the zero initial condition $V(0) = 0$. Hence, its first two derivatives at $t=0$ are given by
\begin{equation}
\label{Vder}
  \dot{V}(0) = \Lambda(0),
  \qquad
  \ddot{V}(0) = A \Lambda(0) + \Lambda(0) A^\rT + \dot{\Lambda}(0),
\end{equation}
where
\begin{equation}
\label{Lamdot}
  \dot{\Lambda}(0)
  =
  4\mho^\rT ((\beta\cdot (A\mu(0)+b))\ox  (M^\rT\Omega M))\mho
\end{equation}
in view of (\ref{EBB}) and since the mean vector (\ref{mu}) satisfies the ODE $\dot{\mu} = A\mu + b$.
By substituting (\ref{FF})--(\ref{P}) into (\ref{Del}), the mean square deviation takes the form
\begin{align}
\nonumber
    \Delta(t)
    = &
    \|F(\re^{tA}-I_n)
    \sqrt{P}\| ^2
    +
    2b^\rT \psi(t)^\rT\Sigma (\re^{tA}-I_n)\mu(0)\\
\label{ExiTxi}
    & +
    |F\psi(t)b|^2
    +
    \bra \Sigma, \Re V(t)\ket .
\end{align}
The Taylor series expansion of $\Delta$, truncated to its quadratic part, is  given by
\begin{equation}
\label{Delasy0}
    \Delta(t)
     =
     \dot{\Delta}(0)
     t +
     \frac{1}{2}
     \ddot{\Delta}(0)
     t^2
     +
     O(t^3),
     \qquad
     {\rm as}\
     t \to 0+,
\end{equation}
where the first two derivatives of $\Delta$ from (\ref{ExiTxi}) at $t=0$  are computed as
\begin{align}
\label{Deldot}
  \dot{\Delta}(0)
  =&
  \bra
    \Sigma, \Re \dot{\Ups}(0)
  \ket
    =
  \bra
    \Sigma, \Re \Lambda(0)
  \ket,\\
\nonumber
   \ddot{\Delta}(0)
   =&
  \bra
    \Sigma, \Re \ddot{\Ups}(0)
  \ket
    =
   2\|F A \sqrt{P}\|^2
   + 4b^\rT \Sigma A \mu(0)
   +
   2|Fb|^2\\
\nonumber
   & +
      \bra
        \Sigma,
        A \Re \Lambda(0) + \Re \Lambda(0) A^\rT + \Re \dot{\Lambda}(0)
   \ket\\
\label{Delddot}
    = &
      \bra
        \Sigma,
        2 APA^\rT +
        2A( \Re \Lambda(0) +2\mu(0)b^\rT) + \Re \dot{\Lambda}(0)
   \ket
   +
   2|Fb|^2
\end{align}
by using (\ref{psi}), (\ref{Vder}), (\ref{Lamdot}). While the short-time expansion (\ref{Delasy0}) along with the relations (\ref{Deldot}), (\ref{Delddot}) are valid regardless of the spectral properties of the matrix $A$, the asymptotic behaviour of (\ref{ExiTxi}) at the other extreme, as $t\to +\infty$, depends on whether the system is dissipative. More precisely, if $A$ is Hurwitz, then $\lim_{t\to +\infty}\re^{tA} = 0$ and $\lim_{t\to +\infty}\psi(t) = -A^{-1}$ in (\ref{psi}),  and hence, (\ref{mu*}), (\ref{V}), (\ref{ExiTxi}) imply that
$
    \lim_{t\to +\infty}
    \Delta(t)
    =
    \|F\sqrt{P}\| ^2
    -
    2\mu_\infty^\rT\Sigma \mu(0)
    +
    |F\mu_\infty|^2
    +
    \bra \Sigma, P_\infty\ket
$.
Here,
$
    P_\infty
    :=
    \int_0^{+\infty}
    \re^{tA}
    \Re \Lambda_\infty
    \re^{tA^\rT}
    \rd t
$
is the infinite-horizon controllability Gramian of the pair $(A,\sqrt{\Re \Lambda_\infty})$ satisfying the algebraic Lyapunov equation (ALE) $AP_\infty + P_\infty A^\rT + \Re \Lambda_\infty = 0$ obtained from a steady-state version of (\ref{VODE}) using the limit
$
    \Lambda_\infty
    :=
    \lim_{t\to +\infty}
    \Lambda(t)
    =
4\mho^\rT ((\alpha + \beta\cdot \mu_\infty)\ox  (M^\rT\Omega M))\mho
$ of (\ref{EBB}).

\section{QUANTUM MEMORY DECOHERENCE TIME}
\label{sec:time}

The quantum system performance in retaining the initial variables   can be described in terms of a ``memory decoherence'' time during which the system variables do not deviate too far from their initial values. Similarly to \cite{VP_2023_ANZCC},  the weighted mean-square deviation approach of (\ref{Del}) suggests   such time to be defined as
\begin{equation}
\label{tau}
    \tau(\eps)
    :=
    \inf
    \big\{t\> 0:\
    \Delta(t)
    >  \eps \|F\sqrt{P}\| ^2
    \big\},
\end{equation}
with the convention $\inf \emptyset := +\infty$. Here, $\eps>0$ is a small dimensionless parameter specifying the relative error threshold for the mean-square deviation $\Delta(t)$  of $X(t)$ from $X(0)$ with respect to a reference quantity
\begin{equation}
\label{ref}
    \bE (X(0)^\rT \Sigma X(0)) = \bra \Sigma, P\ket = \|F\sqrt{P}\|^2,
\end{equation}
which uses (\ref{EXX0}), (\ref{P}). To eliminate from consideration the trivial case of $F\sqrt{P} = 0$ (which is possible if (\ref{FF}) holds with $\nu< n$, when the matrix $F$ has linearly dependent columns), we assume that
\begin{equation}
\label{FP}
  F \sqrt{P}\ne 0.
\end{equation}
The memory decoherence time $\tau(\eps)$  in (\ref{tau})  is a nondecreasing function of the fidelity parameter $\eps>0$ satisfying $\tau(0):= \lim_{\eps\to 0+}\tau(\eps) =0$. Furthermore, $\tau(\eps)>0$ for any $\eps>0$ since $\Delta(t)$ is a continuous function of  $t\> 0$,  with $\Delta(0) = 0$. 
From
(\ref{Delasy0}), it follows  that $\tau(\eps)$ is asymptotically linear, as $\eps\to 0+$,  and its right-hand  derivative at $\eps=0$ is given by
\begin{equation}
\label{ratio}
    \tau'(0)
    :=
    \lim_{\eps\to 0+}
    \frac{\tau(\eps)}{\eps}
    =
    \frac{\|F\sqrt{P}\| ^2}{\dot{\Delta}(0)}
    =
    \frac{\|F\sqrt{P}\| ^2}{  \bra\Sigma, \Re\Lambda(0)\ket}.
\end{equation}
Here, the denominator $\dot{\Delta}(0)$, which is found in (\ref{Deldot}) and is always nonnegative,  is also assumed to be nonzero in what follows:
\begin{equation}
\label{pos}
  \bra\Sigma, \Re\Lambda(0)\ket >0.
\end{equation}
The quantity $\tau'(0)$ resembles the signal-to-noise ratio since the numerator in (\ref{ratio}) pertains to the initial condition $X(0)$ as a useful ``signal'' to be stored (see (\ref{ref})), while the denominator involves the diffusion matrix $\Lambda(0)$  from (\ref{EBB}) associated with the quantum noise. However, note that $\Lambda(0)$ depends on $P$ as well  and that $\tau'(0)$  in (\ref{ratio}) has the physical dimension of time. By using both leading terms from (\ref{Delasy0}), the asymptotic relation (\ref{ratio}) can now be extended to
\begin{equation}
\label{tau12}
    \tau(\eps) =
      \wh{\tau}(\eps)
     + O(\eps^3),
     \qquad
       \wh{\tau}(\eps)
  :=
  \tau'(0) \eps + \frac{1}{2}\tau''(0) \eps^2,
\end{equation}
as $\eps \to 0+$, where the second-order right-hand derivative $\tau''(0)$  of $\tau$ at $\eps=0$ is obtained by matching the truncated expansions as
\begin{equation}
\label{tau2}
    \tau''(0)
      =
    -
    \frac{\ddot{\Delta}(0) \tau'(0)^2}
    {\dot{\Delta}(0)}
    =
    -
    \frac{\ddot{\Delta}(0) \|F\sqrt{P}\|^4}
    {\bra \Sigma, \Re \Lambda(0)\ket^3}   ,
\end{equation}
with $\ddot{\Delta}(0)$ given by (\ref{Delddot}). In contrast to $\tau'(0)$ in (\ref{ratio}), the quantity $\tau''(0)$ depends not only on the coupling matrix $M$, but also  on the energy and coupling vectors  $E$, $N$ from (\ref{HE}), (\ref{LMN})   through the matrix $A$ in (\ref{A_b}) which enters (\ref{tau2}) only through $\ddot{\Delta}(0)$ from (\ref{Delddot}).

A relevant performance criterion    for the quantum memory system   is provided by the maximization of the decoherence time (\ref{tau}) (or its second-order approximation $\wh{\tau}(\eps)$
in (\ref{tau12})) at a fixed fidelity level $\eps$:
\begin{equation}
\label{taumax}
  \tau(\eps)\longrightarrow \sup.
\end{equation}
The resulting optimization problem is over those energy and coupling parameters $E$, $M$, $N$ of the system which are allowed to be varied in a particular setup.
The matrix $F$ (and hence, $\Sigma$) in (\ref{FF}) can be an additional parameter over which $\tau(\eps)$ is maximized. For example, $F$ can be varied so as to find a particular subset  of the system variables $X_1, \ldots, X_n$ (or their linear combinations) which are retained with the given relative accuracy $\eps$  for a longer period of time (\ref{tau}) than the others.

\section{DECOHERENCE TIME SUBOPTIMIZATION}
\label{sec:subopt}

Since the memory decoherence time $\tau(\eps)$ in (\ref{tau}) is a complicated function of the energy, coupling and weighting parameters, we will consider an  
approximate version of the optimization problem (\ref{taumax}),    which makes advantage of the smallness of $\eps$ and the following convexity properties with respect to  the energy vector $E$.
Note that the second-order approximation $\wh{\tau}(\eps)$ of  $\tau(\eps)$ in (\ref{tau12}) depends on the matrix $A$ in a concave quadratic fashion, inheriting this property from $\tau''(0)$ in (\ref{tau2}) since $\ddot{\Delta}(0)$ in (\ref{Delddot}) is a convex quadratic function of $A$. Due to $A$ in (\ref{A_b}) depending affinely on the energy vector   $E$ in view of  (\ref{A0}), (\ref{At}), $\ddot{\Delta}(0)$ is a convex quadratic function of $E$. Therefore, at least asymptotically,     for small values of $\eps$,  the maximization (\ref{taumax}) tends to favour ``localized'' values  of $E$ as specified  by a suboptimal solution of this problem below.

\begin{thm}
\label{th:R}
Suppose the fidelity level $\eps$ in (\ref{tau}) and the  weighting matrix $\Sigma$  in (\ref{FF}) are fixed together with the matrix $P$ in (\ref{P}), and the conditions  (\ref{FP}), (\ref{pos}) are satisfied.  Then for any given coupling parameters $M$, $N$ of the system (\ref{dXdY}) in (\ref{LMN}),  the energy vector $E$ in (\ref{HE})  delivers a solution to the problem
\begin{equation}
\label{tauhatmax}
  \wh{\tau} \longrightarrow \sup,
  \qquad
  E \in \mR^n,
\end{equation}
of maximizing the approximate decoherence time $\wh{\tau}(\eps)$ in
(\ref{tau12}) if and only if
\begin{equation}
\label{REK}
  2RE + K = 0.
\end{equation}
Here, $0\preccurlyeq R = R^\rT \in \mR^{n\x n}$ and $K\in \mR^n$ are auxiliary matrix and vector computed as
\begin{align}
\label{R}
  R  :=& \mho^\rT(P\ox \Sigma) \mho, \\
\nonumber
    K
    := &
   \mho^\rT
   \col
    \Big(\Sigma
    \Big(\wt{A}P + \frac{1}{2} \Re \Lambda(0) +b\mu(0)^\rT
    \Big)\Big)\\
\label{K}
    &+
   (\bra
        \mho \Sigma \mho^\rT,
         \Re ((\beta\cdot (\Theta \cdot \mu(0))_{\bullet k} )\ox  (M^\rT\Omega M))
   \ket)_{1\< k\< n}
\end{align}
using the structure constants from (\ref{Xalg}), the CCR array $\Theta$ in   (\ref{Theta})   and the matrix $\mho$ from (\ref{colA0}) along with (\ref{A_b}), (\ref{At}),  (\ref{EBB}).\hfill$\square$
\end{thm}
\begin{proof}
Since the coefficient $\tau'(0)$ in (\ref{ratio}) does not depend on $E$, then $\sup_{E  \in \mR^n} \wh{\tau}(\eps) = \tau'(0) \eps + \frac{1}{2} \eps^2 \sup_{E\in \mR^n} \tau''(0)$ in view of (\ref{tau12}). This makes (\ref{tauhatmax}) equivalent to maximizing $\tau''(0)$ in (\ref{tau2}) over $E$, which reduces to
\begin{equation}
\label{Delmin}
  \ddot{\Delta}(0)
  \longrightarrow \inf,
  \qquad
  E\in \mR^n.
\end{equation}
From the convex dependence of $\ddot{\Delta}(0)$ on the matrix $A$ in (\ref{Delddot}) and the affine dependence of $A$ on $E$ in (\ref{A_b}) through (\ref{A0}), it follows that (\ref{Delmin}) is a convex minimization problem. Hence, the first-order optimality condition \begin{equation}
\label{opt}
    \d_E \ddot{\Delta}(0) = 0
\end{equation}
on the gradient of $\ddot{\Delta}(0)$ with respect to the energy vector $E:= (E_k)_{1\< k\< n}$
is necessary and sufficient for $E$ to deliver a global  minimum to $\ddot{\Delta}(0)$. We will now compute this gradient by differentiating the terms on the right-hand side of (\ref{Delddot}).  Since the matrix $\wt{A}$ in (\ref{At}) does not depend on $E$, then the first variation of the matrix $A$ in (\ref{A_b}) with respect to $E$ reduces to that of the matrix $A_0$ in (\ref{A0}) as $\delta A = \delta A_0 = 2\Theta \diam \delta E$. Hence,
\begin{align}
\nonumber
    \delta
    \bra
        \Sigma,
        APA^\rT
    \ket
      & =
    \bra
        \Sigma,
        (\delta A_0)PA^\rT + A P\delta A_0^\rT
    \ket
    =
    2
    \bra
        \Sigma AP,
        \delta A_0
    \ket\\
\nonumber
    & =
    2
        \col(\Sigma AP)^\rT
        \delta \col(A_0)\\
\label{deltaAA}
        & =
    4
        \col(\Sigma AP)^\rT
        \mho \delta E
\end{align}
in view of the vectorization (\ref{colA0}). The relation (\ref{deltaAA}), combined with the first equality in (\ref{A_b}),
implies that
\begin{align}
\nonumber
        \d_E
        \bra
        \Sigma,
        APA^\rT
    \ket
    & =
    4 \mho^\rT \col(\Sigma AP)\\
\nonumber
    & =
    4 \mho^\rT ((P\ox \Sigma) \col(A_0) +  \col(\Sigma \wt{A}P))\\
\label{dEAA}
    & =
    4  (2R E  +  \mho^\rT \col(\Sigma \wt{A}P)),
\end{align}
where $R$ is the matrix from (\ref{R}).
Similarly to (\ref{deltaAA}), since the matrix $\Lambda(0)= 4\mho^\rT ((\alpha + \beta\cdot \mu(0))\ox  (M^\rT\Omega M))\mho$ in (\ref{EBB}) and the vectors $b$ in (\ref{A_b}) and $\mu(0)$ in (\ref{mu})   do not depend on $E$, then
$
    \delta
      \bra
        \Sigma,
        A ( \Re \Lambda(0) +2\mu(0)b^\rT)
   \ket
    =
      \bra
        \Sigma ( \Re \Lambda(0) +2b\mu(0)^\rT),
        \delta A_0
   \ket
    =
   2
   \col(\Sigma ( \Re \Lambda(0) +2b\mu(0)^\rT))^\rT
   \mho
   \delta E
$,
which yields
\begin{equation}
\label{dEA}
    \d_E
      \bra
        \Sigma,
        A ( \Re \Lambda(0) +2\mu(0)b^\rT)
   \ket
   =
   2
   \mho^\rT \col(\Sigma ( \Re \Lambda(0) +2b\mu(0)^\rT)).
\end{equation}
From the relation  $A_0 \mu(0) = 2(\Theta \diam E) \mu(0) = 2(\Theta \cdot \mu(0))E$, obtained  by applying the identity (\ref{cdotdiam}) to (\ref{A0}), it follows that
\begin{equation}
\label{dEkAmu}
    \d_{E_k} (A \mu(0))
    =
    \d_{E_k} (A_0 \mu(0))
    =
     2(\Theta \cdot \mu(0))_{\bullet k}
\end{equation}
for all $k = 1, \ldots, n$. A combination of (\ref{dEkAmu}) with
(\ref{Lamdot}) leads to
$
    \d_{E_k} \dot{\Lambda}(0)
     =
    4\mho^\rT ((\beta\cdot \d_{E_k}(A_0\mu(0)))\ox  (M^\rT\Omega M))\mho
     =
    8\mho^\rT ((\beta\cdot (\Theta \cdot \mu(0))_{\bullet k} )\ox  (M^\rT\Omega M))\mho
$,
and hence,
\begin{equation}
\label{dE1}
    \d_E
      \bra
        \Sigma,
        \Re \dot{\Lambda}(0)
   \ket
   =
   8
   (\bra
        \mho \Sigma \mho^\rT,
         \Re ((\beta\cdot (\Theta \cdot \mu(0))_{\bullet k} )\ox  (M^\rT\Omega M))
   \ket)_{1\< k\< n}.
\end{equation}
Since the term $|Fb|^2$ in (\ref{Delddot}) is independent of $E$, then by assembling (\ref{dEAA}), (\ref{dEA}), (\ref{dE1}), it follows that
\begin{align}
\nonumber
    \d_E \ddot{\Delta}(0)
    =&
    8  (2R E  +  \mho^\rT \col(\Sigma \wt{A}P))\\
\nonumber
    & +
       4
   \mho^\rT \col(\Sigma ( \Re \Lambda(0) +2b\mu(0)^\rT))\\
\nonumber
    & +
   8
   (\bra
        \mho \Sigma \mho^\rT,
         \Re ((\beta\cdot (\Theta \cdot \mu(0))_{\bullet k} )\ox  (M^\rT\Omega M))
   \ket)_{1\< k\< n}       \\
\label{8REK}
    =&
    8(2RE + K),
\end{align}
where use is made of (\ref{R}), (\ref{K}). In view of (\ref{8REK}), the condition (\ref{opt}) is equivalent to (\ref{REK}), thus establishing the latter as a necessary and sufficient condition of optimality for the problem (\ref{tauhatmax}).
\end{proof}

By (\ref{8REK}), the matrix $R$ in (\ref{R}) is related to the Hessian matrix of the convex quadratic function $E\mapsto \ddot{\Delta}(0)$ as $\d_E^2 \ddot{\Delta}(0) = 16R$.
In view of (\ref{REK}), the zero energy vector $E=0$ is an optimal solution of the problem (\ref{tauhatmax}) (and thus a suboptimal solution of (\ref{taumax}) for small values of $\eps$) if and only if the vector $K$ in (\ref{K}) vanishes.
The condition $K=0$ is a quadratic constraint on the coupling parameters $M$, $N$ under which $E=0$ is beneficial for maximizing the memory decoherence time of the system in the framework of the approximation $\tau(\eps)\approx \wh{\tau}(\eps)$ in  (\ref{tau12}). Beyond the zero-Hamiltonian case, if the matrix $R$ in (\ref{R}) is positive definite, then the optimal  value of $E$ in the problem (\ref{tauhatmax}) is uniquely found from (\ref{REK}) as $E = -\frac{1}{2}R^{-1} K$. The above setting provides one of possible decoherence time control formulations which can also be considered for interconnections of finite-level systems,  
similar to those of open quantum harmonic oscillators in \cite{VP_2023_SCL,VP_2023_ANZCC}.

\section{DECOHERENCE TIME CONTROL BY DIRECT ENERGY COUPLING}
\label{sec:two}

The decoherence time optimization for a quantum system with an algebraic structure is applicable to interconnections of such systems.
As an illustration, consider two systems (for example, interpreted as a plant and a controller) from \cite[Section~9]{VP_2022},  which have a direct energy coupling between them and interact with external bosonic fields;  see
Fig.~\ref{fig:system}.
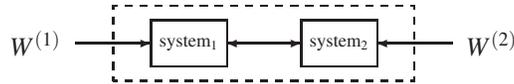
\begin{figure}[htbp]
\centering
\unitlength=1mm
\linethickness{0.5pt}
\begin{picture}(50.00,8)
    \put(5,-2){\dashbox(40,10)[cc]{}}
    \put(10,0){\framebox(10,6)[cc]{\scriptsize system${}_1$}}
    \put(30,0){\framebox(10,6)[cc]{\scriptsize system${}_2$}}
    \put(0,3){\vector(1,0){10}}
    \put(50,3){\vector(-1,0){10}}
    \put(20,3){\vector(1,0){10}}
    \put(30,3){\vector(-1,0){10}}
    \put(-2,3){\makebox(0,0)[rc]{$W^{(1)}$}}
    \put(52,3){\makebox(0,0)[lc]{$W^{(2)}$}}
\end{picture}
\caption{
    An interconnection of two quantum systems, 
    which have external input quantum Wiener processes  $W^{(1)}$, $W^{(2)}$  and interact with each other through a direct  energy coupling.
}
\label{fig:system}
\end{figure}
The external fields  are modelled by
quantum Wiener processes  $W^{(1)}$, $W^{(2)}$ (of even dimensions $m_1$, $m_2$)  on symmetric Fock spaces $\fF_1$, $\fF_2$, respectively. They form an augmented quantum Wiener process
$    W
    :=
    {\scriptsize\begin{bmatrix}
      W^{(1)}\\
      W^{(2)}
    \end{bmatrix}}
$ of dimension $m:= m_1+m_2$
on the composite Fock space $\fF:= \fF_1\ox \fF_2$ with the quantum Ito matrix $    \Omega
    =
    {\scriptsize\begin{bmatrix}
      \Omega_1 & 0\\
      0 & \Omega_2
    \end{bmatrix}}$ in (\ref{Omega}) and the individual Ito tables $
    \rd W^{(s)}\rd W^{(s)}{}^{\rT} = \Omega_s \rd t$. Here,
$
    \Omega_s
     := I_{m_s} + iJ_s$ and
$
    J_s:= I_{m_s/2} \ox \bJ
$, 
with the matrix
$\bJ$ from (\ref{bJ}), so that
$
    J=
    {\scriptsize\begin{bmatrix}
      J_1 & 0\\
      0 & J_2
    \end{bmatrix}}$ in accordance with $W^{(1)}$, $W^{(2)}$ commuting with  each other.
The systems have initial spaces $\fH_s$ and vectors $X^{(s)}$ of $n_s$ dynamic variables on the composite system-field space $\fH:= \fH_0\ox \fF$, with $\fH_0:= \fH_1 \ox \fH_2$ (so that $[X^{(1)}, X^{(2)}{}^\rT] = 0$) and the algebraic structure
\begin{equation}
\label{XXk}
    X^{(s)}X^{(s)\rT}
    =
    \alpha^{(s)} + \beta^{(s)} \cdot X^{(s)},
    \qquad
    s = 1, 2,
\end{equation}
along with the CCR arrays $\Theta^{(s)}:= \Im \beta^{(s)} $.
The system interconnection is endowed with an augmented vector $X$ of $n:= n_1 + n_2 + n_1n_2$ quantum variables (cf. \cite[Example~2]{EMPUJ_2016}):
\begin{equation}
\label{Xvec}
    X:=
    {\begin{bmatrix}
      X^{(1)}\\
      X^{(2)}\\
      X^{(12)}
    \end{bmatrix}},
    \qquad
    X^{(12)} := X^{(1)}\ox X^{(2)},
\end{equation}
which has an algebraic structure (\ref{Xalg}), with the joint structure constants in $\alpha$, $\beta$   computed in \cite[Lemma 9.1]{VP_2022}  in terms of the individual constants from (\ref{XXk}), so that, for example,     $\alpha
    =
    {\scriptsize\begin{bmatrix}
      \alpha^{(1)} & 0 & 0\\
      0 & \alpha^{(2)} & 0\\
      0 & 0 & \alpha^{(1)}\ox \alpha^{(2)}
    \end{bmatrix}}$. Accordingly, the augmented energy vector is partitioned as
\begin{equation}
\label{EEE}
    E:=
    {\begin{bmatrix}
      E^{(1)}\\
      E^{(2)}\\
      E^{(12)}
    \end{bmatrix}},
\end{equation}
where $E^{(s)} \in \mR^{n_s}$ are the individual energy vectors, and $E^{(12)} \in \mR^{n_1n_2}$ parameterizes the  direct coupling Hamiltonian $H_{12}:= E^{(12)}{}^\rT X^{(12)}$ in the total Hamiltonian $H:= H_1 + H_2 + H_{12}$, where $H_s:= E^{(s)}{}^\rT X^{(s)}$. Also, the parameters of coupling of the interconnected system to the augmented quantum Wiener process $W$ are given by
\begin{equation}
\label{LMN1}
    M:=
    {\begin{bmatrix}
      M^{(1)} & 0 & 0\\
      0 & M^{(2)} & 0
    \end{bmatrix}},
    \qquad
    N
    :=
    {\begin{bmatrix}
      N^{(1)}\\
      N^{(2)}
    \end{bmatrix}}
\end{equation}
in terms of the individual coupling parameters $M^{(s)} \in \mR^{m_s\x n_s}$, $N^{(s)} \in \mR^{m_s}$. The resulting $A$, $b$, $B(X)$ in the first QSDE from (\ref{dXdY}) for the augmented system are computed in \cite[Theorem 9.2]{VP_2022}.

Now, if the energy and coupling parameters of the constituent systems in Fig.~\ref{fig:system} are fixed, while the direct energy coupling vector $E^{(12)}$   can be varied, then the latter can be found by solving the approximate decoherence time maximization problem
\begin{equation}
\label{tauhatmaxR12}
  \wh{\tau} \longrightarrow \sup,
    \qquad
  E^{(12)}\in \mR^{n_1n_2},
\end{equation}
which is a reduced version of (\ref{tauhatmax}),  with the individual energy vectors $E^{(s)}$ in (\ref{EEE}) being fixed. A solution of this decoherence time control problem is provided below as a corollary of Theorem~\ref{th:R}. To this end, we partition the matrix $R$ in (\ref{R}) and the vector $K$ in (\ref{K}) for the augmented system as
\begin{equation}
\label{RK12}
  R :=
  {\begin{bmatrix}
    * & * & * \\
    * & * & * \\
    R_1 & R_2 & R_{12}
  \end{bmatrix}},
  \qquad
  K:=
  {\begin{bmatrix}
    * \\
    * \\
    K_{12}
  \end{bmatrix}}  ,
\end{equation}
where $R_s\in \mR^{n_1n_2\x n_s}$, $0 \preccurlyeq R_{12}= R_{12}^\rT\in  \mR^{n_1n_2\x n_1n_2}$, $K_{12}\in  \mR^{n_1n_2}$, and the ``$*$''s denote irrelevant blocks. They are independent of the energy vector $E$, inheriting this property from $R$, $K$.

\begin{thm}
\label{th:R12}
Suppose the system interconnection in Fig.~\ref{fig:system},  specified by (\ref{XXk})--(\ref{LMN1}),  satisfies the assumptions of Theorem~\ref{th:R}. Then the direct energy coupling vector $E^{(12)}$ in (\ref{EEE}) is a solution of the problem (\ref{tauhatmaxR12}) with the approximate decoherence time $\wh{\tau}$ in
(\ref{tau12}) for the augmented system if and only if
\begin{equation}
\label{E12opt}
  2 R_{12} E^{(12)} + Q = 0.
\end{equation}
Here, $Q \in \mR^{n_1 n_2}$ is an auxiliary vector which is computed as
\begin{equation}
\label{Q}
  Q := K_{12} + 2\sum_{s=1}^2 R_s E^{(s)}
\end{equation}
in terms of (\ref{RK12}). \hfill$\square$ 
\end{thm}
\begin{proof}
The proof follows the lines of the proof of Theorem~\ref{th:R} (including the convexity properties mentioned there) except that the optimality condition (\ref{opt}) is replaced with
\begin{equation}
\label{opt12}
    \d_{E^{(12)}} \ddot{\Delta}(0) = 0.
\end{equation}
The left-hand side of (\ref{opt12}) is an appropriate subvector of the gradient $\d_E \ddot{\Delta}(0)$ computed in (\ref{8REK}), which, in view of (\ref{EEE}), (\ref{RK12}), yields
\begin{equation}
\label{subgrad}
    \frac{1}{8}
    \d_{E^{(12)}} \ddot{\Delta}(0)
     =
    2
  {\begin{bmatrix}
    R_1 & R_2 & R_{12}
  \end{bmatrix}}    E + K_{12}
  =
  2R_{12} E^{(12)} + Q,
\end{equation}
with $Q$ given by (\ref{Q}). A combination of (\ref{opt12}) with (\ref{subgrad}) establishes (\ref{E12opt}) as a necessary and sufficient condition of optimality for $E^{(12)}$ in (\ref{tauhatmaxR12}).
\end{proof}

Note that if the matrix $R_{12}$ in (\ref{RK12}) is positive definite, then the optimal  value of $E^{(12)}$ in the problem (\ref{tauhatmaxR12}) is found from (\ref{E12opt}) uniquely  as $E^{(12)} = -\frac{1}{2}R_{12}^{-1} Q$ and depends on the individual energy vectors $E^{(s)}$ in an affine fashion since so does the vector $Q$ in (\ref{Q}). 

Although Theorem~\ref{th:R12} is concerned with a direct energy coupling of two systems, this approach to memory decoherence time  optimization is extendable to more complicated interconnections of several multiqubit systems involving both direct and field-mediated coupling, which will be discussed elsewhere.

\end{document}